\documentclass[twocolumn]{autart}    
\usepackage{amsmath}
\usepackage{amssymb}
\usepackage{graphicx}
\usepackage{bm,bbm}
\usepackage{placeins} 
\usepackage{subfigure} 
\usepackage{overpic}
\usepackage{algorithm}

\usepackage{algorithmic}
\usepackage{enumitem}
\usepackage{xcolor}
\usepackage{hyperref}

\usepackage[normalem]{ulem}
\usepackage[round]{natbib}
\newcommand{\beq}{\begin{equation}}
	\newcommand{\eeq}{\end{equation}}
\newcommand{\bqa}{\begin{eqnarray}}
	\newcommand{\eqa}{\end{eqnarray}}
\newcommand{\nn}{\nonumber}

\newcommand{\bra}[1]{ \langle{#1} |}
\newcommand{\ket}[1]{ |{#1} \rangle}

\newcommand{\sq}[1]{\left[ {#1} \right]}

\newcommand{\tr}[1]{{\rm Tr}\sq{ {#1} }}

\hypersetup{hidelinks}

\journal{}  

\makeatletter
\def\ps@copyright{%
  \let\@mkboth\@gobbletwo
  \let\@oddhead\@empty
  \let\@evenhead\@empty
  \def\@oddfoot{%
    \parbox[t]{\textwidth}{%
      \centering\normalfont\scriptsize
      \textcopyright~2026.
      This manuscript version is made available under the\\
      \href{https://creativecommons.org/licenses/by-nc-nd/4.0/}
           {CC BY-NC-ND 4.0 license.}\\[3pt]
      \thepage
    }%
  }%
  \let\@evenfoot\@oddfoot
}
\makeatother

\begin{document}

\begin{frontmatter}

\title{Efficient quantum state tomography with two complementary  projective   measurements \thanksref{footnoteinfo}}               

\thanks[footnoteinfo]{This research was supported by the Innovation Program for Quantum Science and Technology (No. 2023ZD0301400) and the National Natural Science Foundation of China (No. 12205219, No. 62503036, No. 62273016).}

\author[College]{Xiang Li\thanksref{equal}},    
\author[Hangzhou,Beijing]{Yong Wang\thanksref{equal}},
\thanks[equal]{Equally contributed to this work.} 
\author[Taiyuan]{Lijun Liu\thanksref{corresponding}}\ead{lljcelia@126.com},  
\author[College,Laboratory]{Yiguang Hong\thanksref{corresponding}}\ead{yghong@iss.ac.cn},
\author[Hangzhou,Beijing,BJ,HZ]{Qing Gao},
\author[College,Laboratory]{Shuming Cheng\thanksref{corresponding}}\ead{drshuming.cheng@gmail.com}
\thanks[corresponding]{Corresponding authors.} 

\address[College]{Department of Control Science and Engineering, Tongji University, Shanghai, 201804, China}          
\address[Hangzhou]{Hangzhou Innovation Institute of Beihang University, Zhejiang, 310051, China}
\address[Beijing]{School of Automation Science and Electrical Engineering, Beihang University, Beijing, 100191, China}
\address[Taiyuan]{Department of Mathematical Sciences, Shanxi Normal University, Taiyuan, 030006, China} 
\address[Laboratory]{State Key Laboratory of Autonomous Intelligent Unmanned Systems, Shanghai Research Institute for Intelligent Autonomous Systems, Tongji University, Shanghai, 201203, China}  
\address[BJ]{Embodied Intelligence Robotics Institute, Beihang University, Beijing, 100191, China}
\address[HZ]{Zhejiang Key Laboratory of Industrial Big Data and Robot Intelligent Systems, Zhejiang, 310051, China}


\begin{abstract}
Quantum state tomography (QST) is of fundamental importance to characterize quantum systems in quantum information processing, but its practical implementation is severely hindered by the exponential scaling of measurement and computational costs. In this paper, we present a novel QST protocol that utilizes Kirkwood-Dirac (KD) quasiprobability to reconstruct quantum states. First, it enables state reconstruction with only two complementary rank-one projective  measurements, thus significantly reducing the measurement cost. Then, a complex logistic regression estimator is proposed to process collected KD data, together with  a projected gradient algorithm to mitigate numerical instability and to accelerate convergence. The product-operator structure of KD quasiprobability is further exploited to reduce the computational cost. Finally, extensive experiments are implemented to confirm the validity of our protocol. Notably, the full reconstruction of randomly generated 15-qubit mixed-state instances can be accomplished within 20 minutes under the GPU implementation. These results suggest a promising route toward scalable QST and benchmarking large-scale quantum systems.
\end{abstract}

\end{frontmatter}
\endNoHyper   
\section{Introduction}

Rapid progress recently achieved in quantum information technology has generated a pressing need for reliable methods to characterize quantum devices~\citep{blume2025}. Quantum system identification, inferring quantum properties from measurement data, provides a cornerstone for this demand and covers quantum state tomography (QST)~\citep{Gebhart2023}, quantum process tomography~\citep{KiShSi2014,yu2020}, and quantum detector tomography~\citep{Endo2021,Xiao2022}. Among these, QST that aims to reconstruct unknown quantum states from measurement statistics is the most fundamental, as it is indispensable to others and underpins wide applications ranging from benchmarking and controlling quantum systems~\citep{Proctor2025}, to verification of quantum teleportation~\citep{Bao2012} and to quantum error mitigation~\citep{Cai2023}.

The standard QST workflow consists of first generating data from measurements and then applying an inversion algorithm to recover quantum states from collected data. However, it faces severe challenges in the cost of measurement and computational complexity. For $n$-qubit systems, conventional measurement schemes require $3^n$ Pauli observables and optimized ones need $\mathcal{O}(rd\log^2 d)$ for rank-$r$ states with $d\!=\!2^n$~\citep{Gross2010}, all scaling exponentially  with $n$ and become impractical at large size. Besides, the computational complexity poses an equal bottleneck. Typical examples of maximum-likelihood estimator (MLE)~\citep{CG-APG} and   linear-regression estimator (LRE)~\citep{PLS-Qi2013, mu2020, guctua2020}   involve computation on matrices with dimensionality scaling to $\mathcal{O}(d^2)$, incurring the overwhelming computational cost. Though gradient-based methods provide tractable alternatives to MLE, they often suffer from slow convergence and numerical instabilities in high dimension~\citep{Sproj1,Sproj2}. The LRE method avoids iterative optimization but tends to incur a loss of
accuracy~\citep{PLS-Qi2013,acharya2019}. These issues heavily limit the utility and scalability of QST.  

To alleviate the above issues, we present a new QST protocol that utilizes Kirkwood-Dirac (KD) quasiprobability~\citep{Kirkwood1933, Dirac1945} to reconstruct quantum states. Importantly, the proposed protocol is measurement-efficient in the sense that two complementary rank-one projective measurements are sufficient to determine the unknown state~\citep{Chaturvedi2006,Arvidsson2024}, rather than exponentially many required by conventional ones~\citep{CG-APG}. Moreover, it does not need any prior structural assumptions about the target state, going beyond methods tailored to states with certain structure, such as compressed sensing for low-rank states~\citep{Hsu2024,Hu2025,wang2024quantum}, permutational invariance~\citep{toth2010}, tensor models for low-entanglement states~\citep{qin2024}, and neural states~\citep{Neugebauer2020,Li2023yong}. Thus, our protocol admits near-minimal experimental reconfiguration, significantly reducing measurement burden for experiments.

We then develop  a complex logistic regression estimator (CLRE), together with a projected gradient descent (PGD) algorithm, to efficiently accomplish the KD-based QST.  Specifically, the former can process negative and/or complex KD quasiprobabilities directly, based on an algorithmic surrogate model, while the latter can mitigate numerical instability and speed up convergence.  Furthermore, the tensor-product structure of KD quasiprobabilities  generated from local measurements  is exploited to reduce the per-iteration computational complexity from $\mathcal{O}(d^4)$ to $\mathcal{O}(d^3)$.

We finally conduct extensive experiments through numerical simulations and on the IBM Qiskit simulator to validate our KD-based protocol. Results demonstrate that the CLRE achieves an empirical speedup of up to 100-fold compared with the standard MLE in the tested benchmarks~\citep{CG-APG}. Notably, it is able to reconstruct randomly generated 15-qubit mixed states within 20 minutes under the GPU implementation, whereas previous works report several hours for 14-qubit arbitrary mixed states~\citep{Hou2016} and over 30 minutes for 12-qubit low-rank states~\citep{Hsu2024}. 

This paper is structured as follows. First, Section~\ref{sec:preliminaries} gives a brief introduction to QST and KD quasiprobability, and Section~\ref{sec:feasibility KD} shows that KD quasiprobability is suitable for QST. Then, Section~\ref{sec:MLE KD} proposes CLRE for the KD-based QST, while Section~\ref{sec:PGD} presents the PGD algorithm, the computational complexity of which is analyzed in Section~\ref{sec:computational complexity}. Finally, Section~\ref{sec:Numerical exp} provides experimental validation, and Section~\ref{sec:Conclusion} concludes with this work.

\textbf{Notation:} A complex matrix $X$ with $m$ rows and $n$ columns belongs to the linear space $\mathbb{C}^{m\times n}$, and its trace is denoted by $\tr{X}$. 
The identity and zero matrices are denoted by $\mathbb{I}$ and $\mathbf{0}$, respectively; their dimensions are dropped for simplicity. 
For a square matrix $X$, the notation $X\succeq\mathbf{0}$ indicates that $X$ is positive semidefinite. The symbol $\otimes$ represents the tensor product. 
A quantum state vector is written in bra-ket notation as $|\phi\rangle$, with $\langle\phi|=(|\phi\rangle)^\dagger$. 
The standard Pauli operators for qubits are $\sigma_x=\left(\begin{smallmatrix}0&1\\1&0\end{smallmatrix}\right)$, 
$\sigma_y=\left(\begin{smallmatrix}0&-i\\i&0\end{smallmatrix}\right)$, and 
$\sigma_z=\left(\begin{smallmatrix}1&0\\0&-1\end{smallmatrix}\right)$.

\section{Preliminaries}\label{sec:preliminaries}

\subsection{Quantum state and measurement}
 The state of any quantum system is described by a density operator $\rho$ on the $d$-dimensional Hilbert space $\mathcal{H}$,  which is positive semidefinite with unit trace,
\begin{equation}
	\rho \in \mathcal{L}:=\left\{\rho \in \mathbb{C}^{d \times d}~|~ \rho \succeq \mathbf{0}, ~\mathrm{Tr}[\rho]= 1\right\}.  \label{eqa:rho}
\end{equation}
It immediately follows that $d^2-1$ parameters are generally required to specify $\rho$. For an $n$-qubit system, the dimension $d=2^n$ grows exponentially with $n$, so does the number of state parameters.

Measurements of quantum observables are essential to recover state information. Specifically, an observable is described by a Hermitian operator $O$, and its measurement process is modeled as a positive operator-valued measure (POVM) $ \{E_i\}_{i=1,\dots,I}$ with $E_i\succeq \mathbf{0}$ and $\sum_{i=1}^{I}E_i=\mathbb{I}_d$. Ideally, repeating this process many times produces statistics described by a probability vector $\bm{p}=(p_1,\dots, p_I)^\top$,  governed by the Born's rule
\begin{equation}
	p_i=\mathrm{Tr}[\rho\,E_i] \in [0, 1],~~~\sum_{i=1}^I p_i=1. \label{probability}
\end{equation}
Further, combining outcome probabilities with their eigenvalues  gives rise to  the expectation of observable $O:=\sum_i\lambda_i\,E_i$ as
\begin{equation}
	 o:=\mathrm{Tr}[\rho\,O]=\sum_i \lambda_i\mathrm{Tr}[\rho E_i]=\sum_i \lambda_i\,p_i. \label{expectation}
\end{equation}
If all measurement elements in a POVM are projectors $\Pi_i$ and satisfy the orthogonality, i.e., for all $i, i^\prime=1,\ldots,I$, $\Pi_i\succeq \mathbf{0},\quad\sum_{i=1}^I \Pi_i=\mathbb{I}_d,\quad\mathrm{and}\quad
\Pi_i\Pi_{i^\prime} = \delta_{ii^\prime}\Pi_i$, then it becomes a projective measurement $\{\Pi_i\}_{i=1,\dots, I}$.

\subsection{Quantum state tomography}

QST aims to recover an unknown state from measurement statistics.  It proceeds in two steps.  First, a measurement strategy is implemented to generate measurement statistics described by probabilities in~(\ref{probability}) and/or expectations in~(\ref{expectation}). In practice, the collected data suffer from noise induced by state preparation and measurement errors. Then, a suitable state estimator is applied to find an estimate whose predicted statistics approximately match the observed data. Formally, this state reconstruction task can be stated as

\newtheorem{probb}{\bf Problem}
\begin{probb}
	Given observed probabilities $\hat{p}_1,\dots, \hat{p}_K$ (or expectations $\hat{o}_1,\dots,\hat{o}_K$),  find a physical state $\tilde\rho^{\star}\in\mathcal{L}$ such that the estimated  $\tilde{p}_k=\tr{\tilde\rho^{\star}\,E_k}$ (or $\tilde{o}_k=\tr{\tilde\rho^{\star}\,O_k}$)  is close to $\hat{p}_k$ (or $\hat{o}_k) $ for $k =1,\ldots,K$. \label{eqa:optimi}
\end{probb}

Whether the above problem admits a unique solution is fundamentally determined by whether the underlying measurement is informationally complete or not. 

\newtheorem{df}{\bf Definition}
\begin{df}[Informational completeness]\label{df:IC}
	Denote by $\mathcal{B}(\mathcal{H})$ the space of bounded linear operators acting on $\mathcal{H}$. An operator set $\{O_k\}$ is said to be informationally complete on $\mathcal{B}(\mathcal{H})$ if it provides a representation of any operator $X\in\mathcal{B}(\mathcal{H})$, and if any two distinct operators $X_1,X_2\in\mathcal{B}(\mathcal{H})$ admit different representations.
	In the QST setting considered in this paper, it reduces to the injectivity of the representation map $\rho\mapsto\{\tr{\rho O_k}\}$ on $\mathcal{L}\subset\mathcal{B}(\mathcal{H})$. Equivalently, for any $\rho_1\neq\rho_2\in\mathcal{L}$, there exists at least one $k$ such that $\tr{\rho_1 O_k}\neq \tr{\rho_2 O_k}$~\citep{prugovevcki1977}.
\end{df}

The informational completeness condition ensures that a unique state can be identified from measurement statistics $\{\tr{\rho O_k}\}_{k=1}^K$, and is distinct from the measurement completeness condition $\sum_i E_i=\mathbb{I}$ imposed for $\{E_i\}$ to be a valid POVM. Informational completeness requires at least $K\geq d^2-1$ linearly independent Hermitian operators $O_k$ to reconstruct any $d$-dimensional quantum state, in addition to the unit-trace constraint imposed in~(\ref{eqa:rho}). A common informationally complete choice for $n$-qubit systems is the Pauli set $O_{\bm{i}}=\bigotimes_{j=1}^n \sigma_{i_j}$ with $\sigma_{i_j}\in\{\mathbb{I},\sigma_x,\sigma_y,\sigma_z\}$, having $4^n=d^2$ observables in total.

Problem~\ref{eqa:optimi} can be tackled as an optimization problem using a suitable state estimator.  A typical example is  MLE which minimizes the cross-entropy between the observed and estimated probability distributions~\citep{Smolin2012, CG-APG}, and the corresponding problem is reformulated as
\begin{probb}
	Given observed probabilities $\hat{p}_1,\dots, \hat{p}_K$ and the associated measurement operators $E_1, \dots, E_K$, the maximum likelihood estimate is obtained from 
	\begin{align}
		\tilde\rho^{\star}&=\underset{\tilde{\rho}\in\mathcal{L}}{\mathrm{argmin}}~f_{\rm MLE}(\tilde{\rho})\nn\\
        &=\underset{\tilde{\rho}\in\mathcal{L}}
        {\mathrm{argmin}}\!\left(\!-\sum_{k=1}^{K}\hat{p}_k\log\,\mathrm{Tr}[\tilde{\rho}\,E_k]\!\right)\!.\label{eqa:optimi mle}
	\end{align}\label{prob:mle}
\end{probb}
This optimization problem can be effectively solved by gradient-based algorithms~\citep{Sproj2,Sproj1,CG-APG} to obtain a (sub)-optimal estimate. 

Finally, the reconstruction accuracy in QST is usually assessed by the Uhlmann-Jozsa fidelity, which is defined as~\citep{fidelity1,fidelity2}
\begin{equation} \label{fidelity}
	F(\rho_1,\rho_2) = \left( \mathrm{Tr}\left[\sqrt{\sqrt{\rho_2}\rho_1\sqrt{\rho_2}}\right] \right)^2
\end{equation}
for two states $\rho_1$ and $\rho_2$.

\newtheorem{rek}{\bf Remark}
\begin{rek}\label{remark1}
 Throughout this paper, we denote $\rho$ the true target state, $\tilde{\rho}$ the estimate during optimization, and $\tilde\rho^{\star}$ the final state estimate. Note that $\tilde{\rho}$ may be unphysical, i.e., $\tilde{\rho}\notin\mathcal{L}$ as defined in~(\ref{eqa:rho}). Also note that $\rho$ is unknown during the estimation process as formulated in Problem~\ref{prob:mle} and is only available for performance evaluation via~(\ref{fidelity}).
\end{rek}

\subsection{Kirkwood-Dirac quasiprobability}

It is not straightforward in quantum theory to generalize measurement probability to a joint probability distribution for simultaneously measuring two observables because of measurement disturbance and operator noncommutativity. Instead, Kirkwood and Dirac introduced~\citep{Kirkwood1933,Dirac1945}
\beq
q_{a b}=\tr{ \rho\, \Pi_a \, \Pi_b},~~~\forall~~a, b\label{eqa:KD}
\eeq
as a quasiprobability representation for joint measurement statistics of two observables $O_1=\sum_a\lambda_a\,\Pi_a$ and $O_2=\sum_b\lambda_b\,\Pi_b$.  It shares many similarities with classical joint probabilities: it has unit-sum 
$\sum_{a, b}q_{a b}=1$,
and admits marginal probability distributions
$p_a=\sum_b q_{a b}\in [0, 1],~p_b=\sum_a q_{a b}\in [0, 1]$.

Notably, $q_{a b}$ could be negative or even complex, violating Kolmogorov's axioms and thus called ``quasiprobability". Its nonpositivity reflects nonclassical quantum features and underlies wide applications in quantum thermodynamics~\citep{Gherardini2024}, post-selected metrology~\citep{Arvidsson-Shukur2020}, and foundational studies~\citep{Budiyono2023quantifying,Liu2025}. In the following, KD quasiprobability distribution is shown to be a measurement-efficient tool for state reconstruction without any prior structural assumptions about the state.

\section{State reconstruction via KD quasiprobability}\label{sec:feasibility KD}

We first show that the KD quasiprobability distribution generated with any two complementary rank-one projective measurements provides a complete representation for all quantum states. 
\newtheorem{thmm}{\bf Theorem}
\begin{thmm} 
	  Consider an arbitrary pair of rank-one projective measurements $\{\Pi_a\}_{a=1}^d$ and $\{\Pi_b\}_{b=1}^d$ where rank-one projectors are explicitly expressed as $\Pi_a= \ket{\psi_a}\bra{\psi_a},\,\,\Pi_b= \ket{\phi_b}\bra{\phi_b}$. Assume further that they are complementary, i.e.,
	\beq
		\tr{\Pi_a\Pi_b} = |\langle \psi_a|\phi_b\rangle|^2 \neq 0,~~\forall~ a, b=1,\dots,d. \label{complemnatry}
	\eeq
  Then, we are able to obtain
	\begin{enumerate}
		\item[(i)]
		Any $\rho$ on the $d$-dimensional Hilbert space admits
		\begin{align}
			\rho=\sum_{a,b=1}^dq_{a b}\, \frac{\Pi_b\Pi_a}{\mathrm{Tr}[\Pi_b\Pi_a]}
		\end{align}
		with KD quasiprobability $q_{ab}$ as state coefficient.
		\item[(ii)] The KD-operator set $\{\Pi_b\Pi_a\}_{a,b=1}^d$ is informationally complete, and KD operators are orthogonal under the Hilbert–Schmidt inner product, i.e.,
		\begin{align}\label{orthogonal}
			\mathrm{Tr}\!\big[\Pi_b\Pi_a\,(\Pi_{b'}\Pi_{a'})^\dagger\big]
			=\delta_{aa'}\delta_{bb'}\,\mathrm{Tr}[\Pi_b\Pi_a],
		\end{align}
		for $a,a^{\prime},b,b^{\prime}=1,\ldots,d$.
	\end{enumerate}
	\label{thm:KD-completeness}
\end{thmm}
\begin{pf} 
The measurement completeness condition give \citep{Chaturvedi2006, Arvidsson2024}
\begin{align}
	\rho=\mathbb{I}\, \rho\, \mathbb{I}&=\left(\sum_{b=1}^d\Pi_b\right)\,\rho\,\left(\sum_{a=1}^d \Pi_a\right)\nonumber\\
    &=\left(\sum_{b=1}^d|\phi_b\rangle\langle\phi_b|\right)\,\rho\,\left(\sum_{a=1}^d |\psi_a\rangle\langle \psi_a|\right)\nonumber\\
	&= \sum_{a, b=1}^d\frac{|\phi_b\rangle\langle{\psi_a}|\langle{\phi_b}|{\psi_a}\rangle}{\langle\psi_a|\phi_b\rangle\langle{\phi_b}|{\psi_a}\rangle} \langle{\phi_b}|\rho|\psi_a\rangle\langle\psi_a|\phi_b\rangle \nonumber\\
	&=\sum_{a, b=1}^{d} q_{a b}\,\frac{\Pi_b\Pi_a}{\mathrm{Tr}[\Pi_b\Pi_a]}.
    \label{eqa:kd representation}
\end{align}
Here the complementarity condition as Eq.~(\ref{complemnatry}) ensures that all denominators $\mathrm{Tr}[\Pi_b\Pi_a]$ are nonzero. It indicates that the KD quasiprobability distribution $q_{ab}$ is a feasible representation for any state $\rho$. 

Suppose then that two distinct states $\rho_1,\rho_2\in\mathcal{L}$ admit the same KD quasiprobability distribution $\{q_{ab}\}$. It follows from~(\ref{eqa:kd representation}) that this distribution gives rise to a unique state, thereby implying $\rho_1=\rho_2$ and contradicting the assumption. This is precisely the injectivity required in Definition~\ref{df:IC}, which establishes the informational completeness of these KD operators. Finally, the orthogonality condition in~(\ref{orthogonal}) follows straightforwardly from the orthogonality of projectors $\Pi_a$ and $\Pi_b$.
\hfill $\square$
\end{pf}

The above theorem establishes that KD quasiprobability provides an alternative representation tool for QST. Specifically, it establishes the theoretical foundation for reconstructing an arbitrary quantum state by using the KD quasiprobability distribution generated with two complementary rank-one projective measurements. Moreover, it is flexible because these two measurements can be chosen arbitrarily under the complementarity condition, making it suitable for QST under practical limitations.

\begin{rek}
It is remarked that generally the KD-operator set $\{\Pi_b\Pi_a\}$ is different from $\{\Pi_a\Pi_b\}$, because $\Pi_a$ and $\Pi_b$ do not necessarily commute.  Note that $\Pi_b\Pi_a$ is used as the KD operator in Theorem~\ref{thm:KD-completeness}, differing in the ordering of $\Pi_a\Pi_b$ in $q_{ab}$ as per~(\ref{eqa:KD}). One can also choose $\Pi_a\Pi_b$ as the KD operator, and this choice does not affect the subsequent analysis. Indeed, these two are related by Hermitian conjugation $(\Pi_b\Pi_a)^\dagger=\Pi_a\Pi_b$, and correspondingly, the KD quasiprobabilities satisfy
$\tr{\rho\,\Pi_a\Pi_b}=\tr{\rho\,(\Pi_b\Pi_a)^\dagger}
=\tr{\rho\,\Pi_b\Pi_a}^*$, with $(\cdot)^*$ denoting the complex conjugation.

\end{rek}

Finally, we rigorously formulate the problem of applying KD quasiprobability to QST, similar to Problem~\ref{eqa:optimi}.
\begin{probb}
	Given observed KD quasiprobabilities $\hat{q}_{1},\dots, \hat{q}_K$, find a physical state $\tilde\rho^{\star}\in\mathcal{L}$ such that the estimated  $\tilde{q}_k=\tr{\tilde\rho^{\star}\Pi_a\Pi_b}$ matches  $\hat{q}_{k}$ for $k=1,\dots, K$. \label{kdqst prob}
\end{probb}
The KD-based formulation allows for an informationally incomplete KD operator set, or a subset of KD coefficients with $K<d^2$. Therefore, the above problem may admit multiple solutions, similar to Problem~\ref{eqa:optimi}. Note that the joint outcome of $\{\Pi_a\}$ and $\{\Pi_b\}$ is recorded in the 2-tuple $(a, b)$ with $a=1,\dots, K_1$ and $b=1,\dots, K_2$.  For simplicity, each pair $(a, b)$ is relabeled by a single index $k=(a-1)K_2+b\in\{1,\ldots,K\}$, with $K=K_1K_2$, so that $k$ uniquely specifies a KD outcome, or equivalently, a KD operator.

\begin{rek}
	It follows from~(\ref{eqa:kd representation}) that it is seemingly straightforward to use KD quasiprobability to reconstruct quantum states. However, this is not the case. First, collected data are noisy in practice, so it is highly possible that the reconstructed state is unphysical. Second, it fails to achieve reliable tomography under limited observations, similar to linear inversion~\citep{PLS-Qi2013}. Thus, proper state estimators and algorithms are needed to solve Problem~\ref{kdqst prob}.
\end{rek}

\section{Complex logistic regression estimator}\label{sec:MLE KD}

Note that the cross-entropy in the MLE objective in~(\ref{eqa:optimi mle}) cannot be directly applied for negative and/or complex KD quasiprobabilities, because the statistical or physical interpretation of MLE is lost, although it could be extended to process negative and complex quantities. Thus, we propose using the CLRE to solve Problem~\ref{kdqst prob}, based on an algorithmic surrogate model.  

Let $\boldsymbol{q}=(q_1,\dots,q_K)^\top$ denote the quasiprobability vector. The mapping in CLRE first decomposes $\boldsymbol{q}$ into the real component ${\rm Re}(\boldsymbol{q})$ and the imaginary component ${\rm Im}(\boldsymbol{q})$, and then concatenates them into a real-valued vector 
\begin{align}
	\boldsymbol{r}=&\left[{\rm Re}(\boldsymbol{q}),{\rm Im}(\boldsymbol{q})\right]^\top \nonumber \\
	=&[{\rm Re}(q_1), \ldots, {\rm Re}(q_K), {\rm Im}(q_1), \ldots, {\rm Im}(q_K)]^\top.  \label{eqa:r}
\end{align}
By further applying the standard softmax mapping $\sigma(\cdot)$ to $\boldsymbol{r}$, we obtain a  $2K$-dimensional probability distribution $\boldsymbol{p}=(p_1,\dots,p_{2K})^\top$ with $p_i=\sigma(r_{i})=\frac{e^{r_i}}{\sum_{i=1}^{2K}e^{r_i}}\in(0,1)$, $i=1,\dots, 2K$. It follows immediately that the mapped vector belongs to the interior of the probability simplex $\Delta^{2K}:=\{\bm{p}\in\mathbb{R}^{2K}:p_i>0,\sum_i p_i=1\}$. 

Given any KD vector $\bm{q}\in\mathbb{C}^{K}$, the complex softmax mapping $\mathcal{M}:\mathbb{C}^{K}\!\mapsto\Delta^{2K}$ thus acts as
\begin{align}
\mathcal{M}(\bm{q})\!=\!\left[\!\frac{e^{{\rm Re}(q_1)}}{N}, \!\dots,\! \frac{e^{{\rm Re}(q_K)}}{N}, \frac{e^{{\rm Im}(q_1)}}{N},\!\dots,\! \frac{e^{{\rm Im}(q_K)}}{N}\!\right]^\top\mkern-15mu=\!\bm{p} \label{mapping2}
\end{align}
with the normalization factor $N=\sum_{k=1}^K\left[e^{{\rm Re}(q_k)}+e^{{\rm Im}(q_k)}\right]$. Correspondingly, Problem~\ref{kdqst prob} is reformulated into
\begin{align}\label{CLRE}
	\tilde\rho^{\star}=\underset{\tilde{\rho}\in\mathcal{L}}{\mathrm{argmin}}~f_{\rm CLRE}(\tilde{\rho})=\underset{\tilde{\rho}\in\mathcal{L}}{\mathrm{argmin}}~\Bigl(-\mathcal{M}(\boldsymbol{{\hat{q}}})^\top\ln\mathcal{M}(\boldsymbol{{\tilde{q}}})\Bigr)
\end{align}
with the observed vector $\boldsymbol{\hat{q}}$ and the estimated $\boldsymbol{\tilde{q}}$. The objective function can be expressed more explicitly as 
\begin{align}
	&f_{\rm CLRE}(\tilde{\rho})\nn\\
	=&-\mathcal{M}(\boldsymbol{{\hat{q}}})^\top\left[\ln\frac{e^{{\rm Re}(\tilde{q}_1)}}{\tilde{N}}, \ldots, \ln\frac{e^{{\rm Im}(\tilde{q}_K)}}{\tilde{N}}\right]^\top \label{eqa:KDmle}\\
	=&-\sum_{a,b}\Bigl[\frac{e^{{\rm Re}(\hat{q}_{ab})}}{\hat{N}}\,{\rm Re}(\tilde{q}_{ab})
	+\frac{e^{{\rm Im}(\hat{q}_{ab})}}{\hat{N}}\,{\rm Im}(\tilde{q}_{ab})\Bigr]+\ln\tilde{N}.\nonumber 
\end{align}

\begin{rek}
	There are many choices for the mapping $\mathcal{M}(\cdot)$ in CLRE to  transform complex-valued quasiprobability vector $\bm{q}$ into a normalized positive probability vector $\bm{p}$. Here, it works as logistic regression which typically employs a normalization function to map real-valued inputs to the range $[0, 1]$~\citep{lavalley2008,sperandei2014}, widely used for classification tasks in machine learning~\citep{sunarya2024}.  Note that this construction can be considered as an algorithmic surrogate model rather than one with clear physical interpretation. 
\end{rek}

Finally, it is proven that the complex softmax mapping $\mathcal{M}$ in~(\ref{mapping2}) is bijective once the common-shift degree of freedom~in~(\ref{eqa:r}) is fixed. This degree of freedom arises because adding a real constant to all components of $\boldsymbol{r}$ leaves the softmax output unchanged, which can be eliminated by fixing the gauge $S=\sum_k[{\rm Re}(q_k)+{\rm Im}(q_k)]$.

\newtheorem{propp}{\bf Proposition}
\begin{propp}\label{proposition}
 	 Let $\mathcal{Q}_S=\{\bm{q}\in\mathbb{C}^{K}:\sum_{k}[{\rm Re}(q_k)+{\rm Im}(q_k)]=S\}.$ For any fixed known real number $S$, the restricted mapping $\mathcal{M}|_{\mathcal{Q}_S}:\mathcal{Q}_S\mapsto\Delta^{2K}$ is bijective. 
\end{propp}
\begin{pf} 
    To prove that the restricted mapping $M|_{Q_S}$ is bijective, it is sufficient to show that a quasiprobability vector $\boldsymbol{q}=(q_1,\ldots,q_K)^\top$  can be  uniquely recovered from any probability vector $\boldsymbol{p}=(p_1,\ldots,p_{2K})^\top$. Indeed, if one introduces a vector $\boldsymbol{r}=(r_1,\ldots,r_{2K})^\top$ satisfying
	\beq
	\sum_{i=1}^{2K} r_i=S, \label{assumption1} 
	\eeq
	and
	\beq
	\frac{e^{r_i}}{\sum_{i=1}^{2K}e^{r_i}} =p_i,~~\forall~~i=1,\ldots,2K, \label{assumption2}
	\eeq
	then it reduces to show that there exists a unique $\bm{r}$ from $\bm{p}$, because the mapping from $\boldsymbol{r}$ to $\boldsymbol{q}$ is trivially bijective following Eq.~(\ref{eqa:r}).
	
	Performing the logarithm on both sides of Eq.~(\ref{assumption2}) yields 
	\begin{equation}
		r_i- \ln\left(\sum_{i=1}^{2K}e^{r_i}\right)=\ln p_i.
	\end{equation}
 	Introducing $C:=\ln(\sum_{i=1}^{2K}e^{r_i})$  as an unknown parameter and summing all $r_i$ gives rise to
	\begin{equation}
		\sum_{i=1}^{2K}r_{i}-2KC=S-2KC=\sum_{i=1}^{2K}\ln p_i,
	\end{equation}
	  and hence $C=(S-\sum_{i=1}^{2K}\ln p_i)/(2K)$.  As a consequence, each component of $\boldsymbol{r}$ is uniquely given by
	\begin{equation}
		r_i=\ln p_i + \frac{1}{2K}\left(S-\sum_{i=1}^{2K}\ln p_i\right).
	\end{equation}
	It immediately follows that the mapping from $\boldsymbol{p}$ to $\boldsymbol{r}$ is one-to-one, completing the proof. \hfill $\square$
\end{pf}

In the informationally complete case with $K=d^2$, the common-shift degree of freedom is automatically fixed by the inherent normalization $S=\sum_{k=1}^{d^2}[{\rm Re}(q_k)+{\rm Im}(q_k)]=1$. 
For $K<d^2$, this gauge-fixing value may be unavailable or not explicitly enforced, so the common-shift information is not retained by $\mathcal{M}$ alone and is not explicitly used in the CLRE objective. Nevertheless, as the available information is itself incomplete in this case, unique reconstruction cannot be guaranteed in general and CLRE still works as a practical heuristic objective to yield a candidate state estimate.

\section{Fast projected gradient descent algorithm}\label{sec:PGD}

We continue to present the PGD algorithm to efficiently solve the optimization problem in~(\ref{CLRE}). In particular, the update rule for state estimate at the $l$-th iteration is
\begin{equation}
	\tilde{\rho}_{l+1} = \mathcal{S}\Bigl[\tilde{\rho}_{l}
	- \alpha \,\nabla f_{\rm CLRE}(\tilde{\rho}_{l})\Bigr].
	\label{eqa:MLE iteration}
\end{equation}
    Here $\alpha$ is the step size and  $\mathcal{S}$ denotes the projection which maps the tentative update to the closest physical state with respect to the Frobenius norm, i.e., $\mathcal{S}[X] := \underset{\bar\rho\in\mathcal{L}}{\mathrm{argmin}}\, \|\bar\rho - X \|_F$~\citep{Sproj1,Smolin2012,Sproj3}. It reduces to the simplex projection of eigenvalues after diagonalizing the Hermitian $X$.  Furthermore, the gradient with respect to the state admits  the following  closed form.

\begin{propp}
	With observed KD quasiprobabilities $\hat{q}_{ab}$, the gradient  of $f_{\rm CLRE}(\tilde{\rho})$  is given as
	\begin{align}
		&\nabla f_{\rm CLRE}(\tilde{\rho})
		=-\frac{1}{2}\sum_{a,b}\Bigl[\Bigl(\frac{e^{{\rm Re}(\hat{q}_{ab})}}{\hat{N}}-\frac{e^{{\rm Re}(\tilde{q}_{ab})}}{\tilde{N}}\Bigr) \nonumber\\
		& \qquad +i\,\Bigl(\frac{e^{{\rm Im}(\hat{q}_{ab})}}{\hat{N}}-\frac{e^{{\rm Im}(\tilde{q}_{ab})}}{\tilde{N}}\Bigr)\Bigr]\,\Pi_b\Pi_a+ \text{h.c.},
		\label{mle gradient KD}
	\end{align}
	where $\tilde{N}$ and $\hat{N}$ are the normalization factors for the  predicted and measured KD, respectively, and $\text{h.c.}$ denotes the Hermitian conjugate of the first term. 
\end{propp}

\begin{pf}
	 Denote by $r_{ab}=\e^{{\rm Re}(\hat{q}_{ab})}/{\hat{N}}$ and $m_{ab}=e^{{\rm Im}(\hat{q}_{ab})}/{\hat{N}}$  the scaled real and imaginary components, respectively,  which remain constant during the update of~(\ref{eqa:MLE iteration}). Thus, the objective function in~(\ref{eqa:KDmle}) is rewritten as
		\begin{align}
			 f_{\rm CLRE}(\tilde{\rho})=-\sum_{a,b}\left[r_{ab}\,{\rm Re}(\tilde{q}_{ab})+m_{ab}\,{\rm Im}(\tilde{q}_{ab})\right] + \ln\tilde{N}. \label{objective}
    	\end{align}
	As $f_{\rm CLRE}$ is a real-valued function of complex matrix $\tilde{\rho}$, the Wirtinger calculus is adopted to compute its gradient and the corresponding Wirtinger gradient is given by $\nabla f_{\rm CLRE}(\tilde{\rho}) :=\nabla_{\tilde\rho^*}f_{\rm CLRE}(\tilde{\rho})= \partial f/\partial \tilde\rho^*$~\citep{li2008,schreier2010}, where $\tilde{\rho}^*$ is the complex conjugate of $\tilde{\rho}$ and the partial derivative is taken over entrywise with $(\partial f/\partial \tilde\rho^*)_{ij}=\partial f/\partial \tilde\rho^*_{ij}=(\partial f/\partial {\rm Re}(\tilde\rho_{ij})\!+\!i\partial f/\partial {\rm Im}(\tilde\rho_{ij}))/2$. Since the equality $\nabla_{\tilde\rho^*}f_{\rm CLRE}(\tilde{\rho})=(\nabla_{\tilde{\rho}}f_{\rm CLRE}(\tilde{\rho}))^*$ holds for any real-valued function, we can first compute $\nabla_{\tilde{\rho}}f_{\rm CLRE}(\tilde{\rho})$ and then obtain the desired Wirtinger gradient by applying the complex conjugate operation.

	Given Hermitian $\tilde{\rho}$, the relations  ${\rm Re}(\tilde{q}_{ab})\!=\!(\mathrm{Tr}[\tilde{\rho}\,\Pi_a\Pi_b]\!+\!\mathrm{Tr}[\tilde{\rho}\,\Pi_b\Pi_a])\!/2$ and ${\rm Im}(\tilde{q}_{ab})=(\mathrm{Tr}[\tilde{\rho}\,\Pi_a\Pi_b]-\mathrm{Tr}[\tilde{\rho}\,\Pi_b\Pi_a])/2i$ lead to
	\begin{align}
		\nabla_{\tilde{\rho}}\, {\rm Re}(\tilde{q}_{ab})&=\frac{1}{2}[(\Pi_a\Pi_b)^\top+(\Pi_b\Pi_a)^\top], \label{gradientRe}\\
		\nabla_{\tilde{\rho}} \,{\rm Im}(\tilde{q}_{ab})&=\frac{1}{2i}[(\Pi_a\Pi_b)^\top-(\Pi_b\Pi_a)^\top]. \label{gradientIm}
	\end{align}
	These equalities follow from the matrix calculus identity $\nabla_A \mathrm{Tr}[AB] = B^\top$  for any complex matrices $A$ and $B$. The gradient of $\tilde{N}\!\!=\!\sum_{a,b}\!\left[e^{{\rm Re}(\tilde{q}_{ab})}+e^{{\rm Im}(\tilde{q}_{ab})}\!\right]$ is obtained as
	\begin{align}
		\nabla_{\tilde{\rho}}\,\tilde{N}&=\,\frac{1}{2}\sum_{a,b}e^{{\rm Re}(\tilde{q}_{ab})}[(\Pi_a\Pi_b)^\top+(\Pi_b\Pi_a)^\top]\nonumber\\
		&+\frac{1}{2i}\sum_{a,b}e^{{\rm Im}(\tilde{q}_{ab})}[(\Pi_a\Pi_b)^\top-(\Pi_b\Pi_a)^\top]. \label{gradientN}
	\end{align}
	Combining Eqs.~(\ref{objective}) to (\ref{gradientN}) immediately yields
	\begin{align}
		&\nabla_{\tilde{\rho}}f_{\rm CLRE}(\tilde{\rho})\nn\\
		&=-\sum_{a,b}\left\{r_{ab}\nabla_{\tilde{\rho}}\, {\rm Re}(\tilde{q}_{ab})+m_{ab}\nabla_{\tilde{\rho}}\, {\rm Im}(\tilde{q}_{ab})\right\}+\frac{\nabla_{\tilde{\rho}}\tilde{N}}{\tilde{N}}\nonumber\\ 
		&=-\frac{1}{2}\!\sum_{a,b}\!\Bigl\{\!\Bigl[(r_{ab}\!-\!\frac{e^{{\rm Re}(\tilde{q}_{ab})}}{\tilde{N}})-i(m_{ab}\!-\!\frac{e^{{\rm Im}(\tilde{q}_{ab})}}{\tilde{N}}\!)\!\Bigr]\Pi_b^\top\Pi_a^\top\nonumber\\
		&~~~+\Bigl[(r_{ab}-\frac{e^{{\rm Re}(\tilde{q}_{ab})}}{\tilde{N}})+i(m_{ab}-\frac{e^{{\rm Im}(\tilde{q}_{ab})}}{\tilde{N}})\Bigr]\Pi_a^\top\Pi_b^\top\Bigr\}.\nn
	\end{align}
	
	Finally, applying the conjugate operation to $\nabla_{\tilde{\rho}}f_{\rm CLRE}(\tilde{\rho})$ yields~(\ref{mle gradient KD}) as desired. \hfill $\square$
\end{pf}

\begin{algorithm}[t]
	\caption{Projected gradient descent algorithm}
	\label{alg:KD_QST}
	\renewcommand{\algorithmicrequire}{\textbf{Input:}}
	\renewcommand{\algorithmicensure}{\textbf{Output:}}
	
	\begin{algorithmic}[1]
		\REQUIRE Observed KD quasiprobabilities $\hat{q}_{ab}$, measurement projectors $\Pi_a$, $\Pi_b$, step size $\alpha$
		\ENSURE State estimate $\tilde{\rho}^{\star}$
		
		\STATE Group all $\hat{q}_{ab}$ into vector $\bm{\hat{q}}$ and compute $\mathcal{M}(\boldsymbol{{\hat{q}}})$ according to~(\ref{mapping2})
		\STATE Initialize $\tilde{\rho}_0$  as a random density matrix 
		
		\FOR{$l$ from 0 to L}
		\STATE Compute the estimated KD vector $\bm{\tilde{q}}$  from $\tilde{\rho}_l$ by~(\ref{eqa:KD})
		\STATE Update $\mathcal{M}(\tilde{\bm{q}})$  by the complex softmax mapping in~(\ref{mapping2})
		\STATE Calculate $\nabla f_{\rm CLRE}(\tilde{\rho}_l)$ according to~(\ref{mle gradient KD})
		\STATE Update state estimate $\tilde{\rho}_{l+1}$ by~(\ref{eqa:MLE iteration})
		\ENDFOR
	\end{algorithmic}
\end{algorithm}

It is evident that the gradient in~(\ref{mle gradient KD}) is Hermitian by computation.  
The complete algorithmic procedure is summarized in Algorithm~\ref{alg:KD_QST} which offers additional computational advantages.  Specifically, conventional gradient-based algorithms for MLE involve evaluating $-\sum_{k}\hat{p}_{k}E_{k}/\tr{\tilde{\rho}E_k}$, which become numerically unstable when some predicted probabilities $\mathrm{Tr}[\tilde{\rho}E_k]$ are very small while empirical probabilities $\hat{p}_{k}$ remain nonzero. This situation typically happens when $\tilde{\rho}$ becomes rank deficient, which could be mitigated by $\epsilon$-flooring~\citep{Blume2010}, damping~\citep{DG-Rehacek2007}, or line search~\citep{Sproj2}.
By contrast, our CLRE avoids the $1/\tr{\tilde{\rho}E_k}$-type ratio by using residual (difference-form) weights uniformly bounded by $ \left|\Bigl(\tfrac{e^{{\rm Re}(\hat{q}_{ab})}}{\hat{N}}\!-\!\tfrac{e^{{\rm Re}(\tilde{q}_{ab})}}{\tilde{N}}\Bigr)\!+\!i\,\Bigl(\tfrac{e^{{\rm Im}(\hat{q}_{ab})}}{\hat{N}}\!-\!\tfrac{e^{{\rm Im}(\tilde{q}_{ab})}}{\tilde{N}}\Bigr)\right|\!<\!\sqrt{2}.$ This makes the update magnitude less sensitive to small predicted probabilities, thus improving numerical stability. In practice, it often admits smoother iterations and faster convergence, allowing for simple step-size schedules with fewer safeguards. It is remarked that this advantage is empirical and may depend on the state rank, noise, and initialization.

\section{Computational complexity of CLRE}\label{sec:computational complexity}

It follows directly from Algorithm~\ref{alg:KD_QST} that the computational complexity of PGD for CLRE  consists of  four components: KD quasiprobability computation, softmax mapping, gradient calculation, and state update. The softmax mapping in~(\ref{mapping2}) is element-wise  and costs  $\mathcal{O}(K)$. The state update in~(\ref{eqa:MLE iteration}) is dominated by the projection step $\mathcal{S}$,  whose  complexity amounts to an eigenvalue decomposition as $\mathcal{O}(d^3)$~\citep{Smolin2012,Wang2024factored}.  Naively, computing $\tr{\tilde\rho\,O}$ needs $\mathcal{O}(d^2)$ multiply-adds, so the complexity of computing $K=d^2$ quasiprobabilities $q_{k}$ is $\mathcal{O}(d^4)$. The same computational complexity can be obtained for gradient computation, by noting that it involves multiplying $d^2$ scalar weights with matrices and summing the total $K$ matrices.

To reduce the computational complexity, we consider systems with an $n$-partite tensor-product structure and homogeneous local dimension $d_{\mathrm{loc}}$. This is an additional assumption introduced only for the complexity reduction and is not required by the general KD reconstruction framework developed in the preceding sections. Under this structure, the tensor-product form of trace computation~\citep{CG-APG} can be exploited to accelerate the calculation of KD quasiprobabilities and gradients. Specifically, any operator on the composite system can be decomposed into a weighted sum of operators $O_w$ in the tensor product structure:
\beq
O_w=O_{w_1}\otimes O_{w_2}\otimes\cdots\otimes O_{w_n}, \label{tensorproduct}
\eeq
where $\{O_{w_s}\}_{w_s=1}^{d_{\rm loc}^2}$ is a local operator basis on the $s$-th subsystem and $s=1,\cdots, n$. If measurement projectors $\Pi_a$s and $\Pi_b$s for KD quasiprobability obey the above structure, these measurements are called local. 

We first consider that local KD operators are informationally complete with $K=d^2$, by showing that the computational complexity of both steps is reduced to $\mathcal{O}(n d^2)$. Particularly, the overall per-iteration cost is dominated by the projection step, which is implemented as a full dense Hermitian eigen-decomposition and scales as $\mathcal{O}(d^3)$. In our implementation, these dense linear-algebra routines are dispatched to parallel GPU compute units, reducing the constant factors and making the dense projection practically feasible in our high-dimensional benchmarks. We further extend the complexity analysis to informationally incomplete and the non-local measurement cases.

\subsection{Quasiprobability calculation with tensor-product structure}\label{sec:recursive-trace}

Denote the computational bases of the $s$-th subsystem by $|j_s\rangle \in\{|0\rangle,|1\rangle,\ldots,|d_{\rm loc}-1\rangle\}$. Any $n$-subsystem state $\rho$ can be  expanded in the tensor product structure as
\beq\label{eqa:rho decomp}
\rho = \mkern-10mu\sum_{\substack{j_1,\dots,j_{n-1}\\j^\prime_1,\ldots,j^\prime_{n-1}}}\mkern-10mu
{|j_1\rangle\langle j^\prime_1|\otimes\cdots\otimes|j_{n-1}\rangle\langle j^\prime_{n-1}|\otimes R_{\bm{j}_n\bm{j}_n^\prime}},
\eeq
where $\bm{j}_n:=(j_1,\dots,j_{n-1})$ and $\bm{j}_n^\prime:=(j_1^\prime,\ldots,j_{n-1}^\prime)$ with $j_s, j^\prime_s=0,\cdots, d_{\rm loc}-1$ and $s=1,\dots, n-1$. Here $R_{\bm{j}_n\bm{j}_n^\prime}$ is a $d_{\rm loc}\times d_{\rm loc}$ matrix on the $n$-th subsystem, defined by the partial trace $\tr{\rho \left(\ket{j^\prime_1}\bra{j_1}\otimes\cdots\otimes\ket{j^\prime_{n-1}}\bra{j_{n-1}} \otimes\mathbb{I}_{d_{\mathrm{loc}}}\right)}$.

The trace $\tr{\rho\,O_w}$ can be efficiently calculated by iteratively exploiting the tensor-product structure. Applying the $n$-th local operator $O_{w_n}$ to state $\rho$ yields a reduced matrix
\begin{align}
	&R^{(w_n)} \mkern-3mu=\mkern-17mu\sum_{\substack{j_1,\dots,j_{n-1}\\j^\prime_1,\ldots,j^\prime_{n-1}}}\mkern-15mu{|j_1\rangle\langle j^\prime_1|\otimes\cdots\otimes|j_{n-1}\rangle\langle j^\prime_{n-1}|\tr{R_{\bm{j}_n\bm{j}_n^\prime}O_{w_n}}} \nn\\
	&=\sum_{\substack{j_1,\dots,j_{n-2}\\j^\prime_1,\ldots,j^\prime_{n-2}}}|j_1\rangle\langle j^\prime_1|\otimes\cdots\otimes|j_{n-2}\rangle\langle j^\prime_{n-2}|\otimes R_{\bm{j}_{n-1}\bm{j}^{\prime}_{n-1}}^{(w_n)}\nn
\end{align}
with $\bm{j}_{n-1}=(j_1,\dots, j_{n-2}) $ and $\bm{j}^\prime_{n-1}=(j_1^\prime,\dots,j^\prime_{n-2})$. The second equality follows similarly from~(\ref{eqa:rho decomp}), and $R_{\bm{j}_{n-1}\bm{j}^{\prime}_{n-1}}^{(w_n)}=\sum_{j_{n-1},j^\prime_{n-1}}|j_{n-1}\rangle\langle j^\prime_{n-1}|\tr{R_{\bm{j}_n\bm{j}^\prime_n}O_{w_n}}$. 

Repeating the same contraction from the $n$-th subsystem to the $s$-th gives
\begin{align}
	&\quad R^{(w_{s}\cdots w_n)} \nn\\
&= \mkern-15mu\sum_{\substack{j_1,\dots,j_{s-1}\\j^\prime_1,\ldots,j^\prime_{s-1}}}\mkern-13mu
|j_1\rangle\langle j^\prime_1|\otimes\cdots\otimes|j_{s-1}\rangle\langle j^\prime_{s-1}|\tr{R_{\bm{j}_{s}\bm{j}^{\prime}_{s}}^{(w_{s+1}\cdots w_n)}O_{w_{s}}} \nn\\
&= \sum_{\substack{j_1,\dots,j_{s-2}\\j^\prime_1,\ldots,j^\prime_{s-2}}}|j_1\rangle\langle j^\prime_1|\otimes\cdots\otimes|j_{s-2}\rangle\langle j^\prime_{s-2}|\otimes R_{\bm{j}_{s-1}\bm{j}^\prime_{s-1}}^{(w_{s}\cdots w_{n})}, \nn
\end{align}
where $\bm{j}_{s}=(j_1,\dots, j_{s-1})$, $\bm{j}^\prime_{s}=(j_1^\prime,\dots,j^\prime_{s-1})$  and $s$ goes from $n-1$ to 1. Here $R_{\bm{j}_{s}\bm{j}^{\prime}_{s}}^{(w_{s+1}\cdots w_n)}=\sum_{j_{s},j^\prime_{s}}|j_{s}\rangle\langle j^\prime_{s}|\cdot\mathrm{Tr}[R^{(w_{s+2}\cdots w_n)}_{\bm{j}_{s+1}\bm{j}^\prime_{s+1}}O_{w_{s+1}}]$. Proceeding this procedure recursively for $n$ steps yields the expectation as desired $R^{(w_1\cdots w_n)}=\tr{\rho\, O_{w_1}\otimes\cdots\otimes O_{w_n}}=\tr{\rho\,O_w}$.

The computational advantage of this procedure lies in reusing $R^{(w_{s+1}\ldots w_n)}$ across all possible $w_s$, hence benefiting for computing all $d^2$ expectations $\tr{\rho\,O_w}$.  Specifically, computing one $R^{(w_s\cdots w_n)}$ from $R^{(w_{s+1}\cdots w_n)}$ requires evaluating $d_{\rm loc}^{2(s-1)}$ local traces, each costing $2d_{\rm loc}^2$ operations, so the total cost is 
$2d_{\rm loc}^{2(s-1)}\cdot d_{\rm loc}^2 = 2d_{\rm loc}^{2s}$. Further, enumerating all possible $O_w$ obtains that the number of $R^{(w_s\cdots w_n)}$ is up to $d_{\rm loc}^{2(n-s+1)}$. As a consequence, aggregating over all steps yields
$\sum_{s=1}^{n}2d_{\rm loc}^{2s}\cdot d_{\rm loc}^{2(n-s+1)}
=2n\,d_{\rm loc}^{2n+2}
=2d_{\rm loc}^2\,nd^2$
with $d=d^n_{\rm loc}$ for quantum systems with $n$ subsystems, leading to the computational complexity $\mathcal{O}(nd^2)$ for a complete KD quasiprobability distribution generated from local measurements.

\subsection{Gradient calculation with tensor-product structure}

The same idea can be applied to calculate the gradient in~(\ref{mle gradient KD}), which amounts to evaluating a weighted sum of operators in tensor-product structure as
\begin{equation}
	O = \sum_{w} c_w O_w = \sum_{w_1,\ldots,w_n} c_{w_1 \ldots w_n} O_{w_1} \otimes \ldots \otimes O_{w_n}.\label{eqa:A decomp}
\end{equation}
First, define partial sums over the last subsystem as $O^{(w_1\ldots w_{n-1})}=\sum_{w_n} c_{w_1\cdots w_n}O_{w_n}.$  Then, recursively extend to include the previous subsystem and define the partial sum operator acting on subsystems $(s+1,\ldots,n)$ as $O^{(w_1\cdots w_s)}=\sum_{w_{s+1}}O_{w_{s+1}} \otimes O^{(w_1\ldots w_{s+1})}$ for $s$ from $n-2$ down to 1. Combining them then gives rise to
\begin{align}\label{eqa:A j1j0}
	&\sum_{w_1} O_{w_1} \otimes O^{(w_1)}=\sum_{w_1,w_2} O_{w_1}\otimes O_{w_2} \otimes O^{(w_1w_2)}\nn\\
	\quad=&\cdots=\sum_{w_1,\cdots,w_n} c_{w_1 \ldots w_n} O_{w_1} \otimes \cdots \otimes O_{w_n}=O. \nn
\end{align}
To compute $O^{(w_1\cdots w_s)}$ at step $s=1,\cdots, n$, the tensor-product operation costs $d_{\rm loc}^{2(n-s+1)}$ and the matrix addition $(d_{\rm loc}^2-1)d_{\rm loc}^{2(n-s)}$. As the total number of $O^{(w_1\ldots w_s)}$ is up to $d_{\rm loc}^{2s}$, the per-step cost is given by $[d_{\rm loc}^{2(n-s+1)}+(d^2_{\rm loc}-1)d_{\rm loc}^{2(n-s)}]d_{\rm loc}^{2s} = (2d^2_{\rm loc}-1)d_{\rm loc}^{2n}.$
Thus, summing over all steps yields the total complexity 
$\sum_s (2d^2_{\rm loc}-1)d_{\rm loc}^{2n}=(2d^2_{\rm loc}-1)nd^{2}$,
also scaling to $\mathcal{O}(nd^2)$.

\subsection{Informationally incomplete and nonlocal cases}

The above technique can be applied to the informational incomplete case or $K<d^2$, by embedding unmeasured KD operators into a complete local KD-operator basis set. Evidently, the complexity $\mathcal{O}(nd^2)$ provides an upper bound for those of $K$ quasiprobabilities and the gradient. Additional savings could be achieved if some additional structure is further exploited, such as sparsity.

For nonlocal measurements, KD operators $O_k$ cannot be factorized as~(\ref{tensorproduct}), but can be expanded as $O_k=\sum_{\bm{w}}\alpha_{k,\bm{w}}\;O_{w_1}\otimes\cdots\otimes O_{w_n}$, where $\bm{w}=(w_1,\ldots,w_n)$. Correspondingly, $\tr{\rho O_k}=\sum_{\bm{w}}\alpha_{k,\bm{w}}\,{\rm Tr}[\rho\,O_{w_1}\otimes\cdots\otimes O_{w_n}]$. The $d^2$ basis terms ${\rm Tr}[\rho\,O_{w_1}\otimes\cdots\otimes O_{w_n}]$ can be computed by the above iterative procedure at a cost of $\mathcal{O}(nd^2)$.
Proceeding with $K$ weighted sums then costs $\sum_{k=1}^{K}s_k$, with $s_k:=\big|\{\bm{w}:\alpha_{k,\bm{w}}\neq 0\}\big|\in\{1,\ldots,d^2\}$ being the number of nonzero coefficients.
In the worst case (dense expansions), it recovers the $\mathcal{O}(d^4)$ scaling. The same idea also applies to the gradient term.

\section{Experiments and results}\label{sec:Numerical exp}

We implement a comprehensive experimental evaluation of our KD-based QST protocol by exploring algorithm convergence, computational efficiency, measurement overhead, and noise robustness. Our results demonstrate that the CLRE estimator with the PGD algorithm achieves fast state reconstruction with high accuracy, speeding up by approximately two orders of magnitude compared to the runtime of Pauli-based MLE in the tested benchmarks. We also examine the measurement resource consumption of the KD-based protocol and investigate its noise robustness on the Qiskit platform.

\subsection{Experimental setup}\label{sec:exp setup}

The $n$-qubit state $\rho$ with rank $r$ is randomly generated according to the Hilbert–Schmidt measure that recovers the Haar measure in the pure-state limit~\citep{Kumar2020}
\begin{equation}
\rho =  \frac{UU^{\dagger}}{\tr{UU^{\dagger}}},~\text{where}~U =U_{\rm Re}+iU_{\rm Im}\in\mathbb{C}^{d\times r}. \label{mixed state}
\end{equation}
The matrix elements of $U_{\rm Re}$ and $U_{\rm Im}$  are independently drawn from a standard normal distribution. The rank $r$ is uniformly sampled from $1$ to $2^n$ to comprehensively evaluate the protocol across states of varying purity, from pure to maximally mixed configurations.

The KD measurement operators used in this work are chosen as two mutually unbiased sets $\{\Pi_a\}=\{|0\rangle\langle{0}|, |{1}\rangle\langle{1}|\}^{\otimes n}$, and
$\{\Pi_b\}=\{|{+}\rangle\langle{+}|, |{-}\rangle\langle{-}|\}^{\otimes n}$, where $|0\rangle=(1,0)^\top$ and $|1\rangle=(0,1)^\top$ are eigenstates of $\sigma_z$, while $|{+}\rangle=1/\sqrt{2}(1,1)^\top$ and $|{-}\rangle=1/\sqrt{2}(1,-1)^\top$ are eigenstates of $\sigma_x$. 

\begin{figure}[t]
	\begin{center}
		\begin{overpic}[width=0.70\columnwidth]{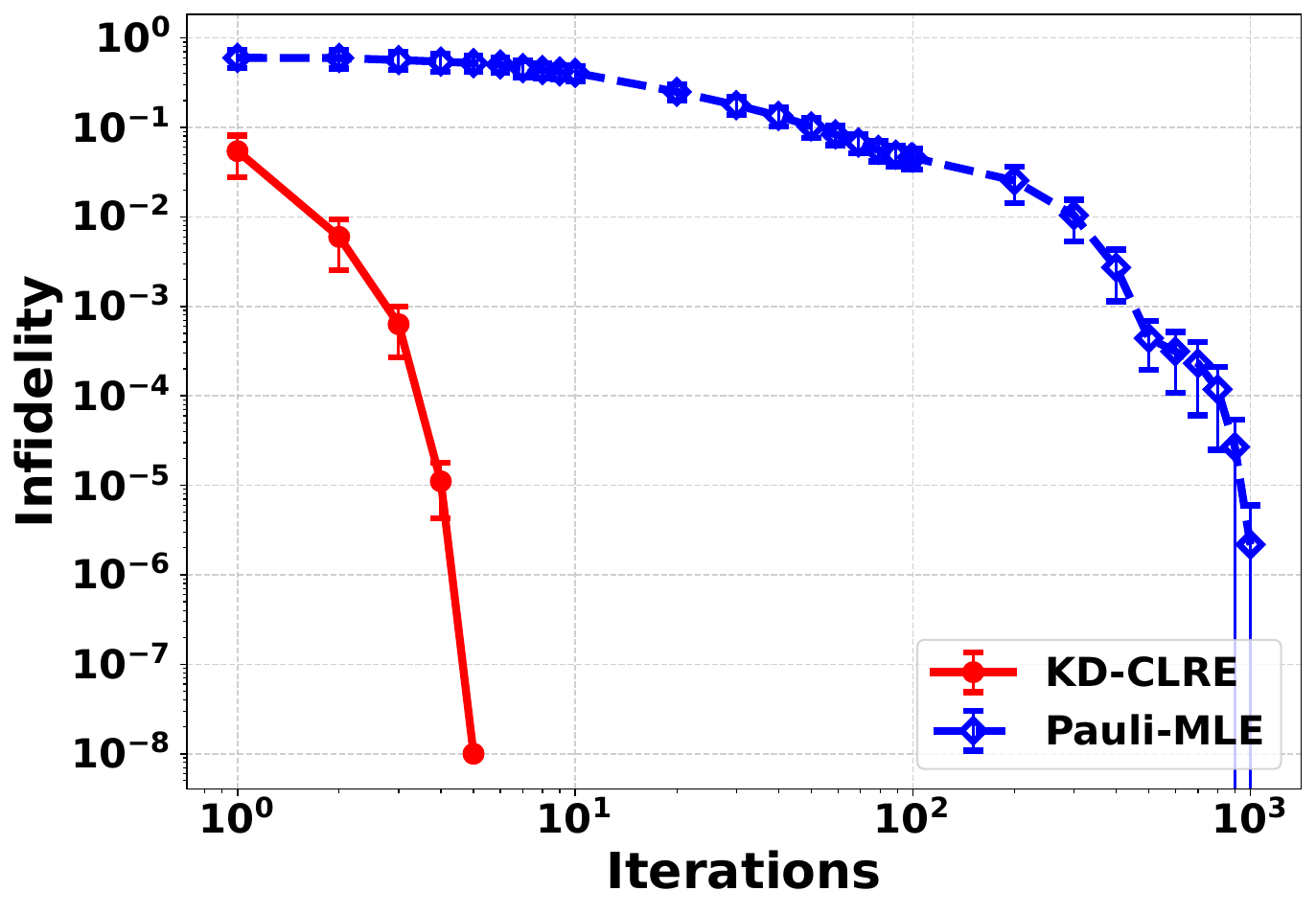}
			\put(-8,65){\textbf{(a)}}
		\end{overpic}
		\begin{overpic}[width=0.70\columnwidth]{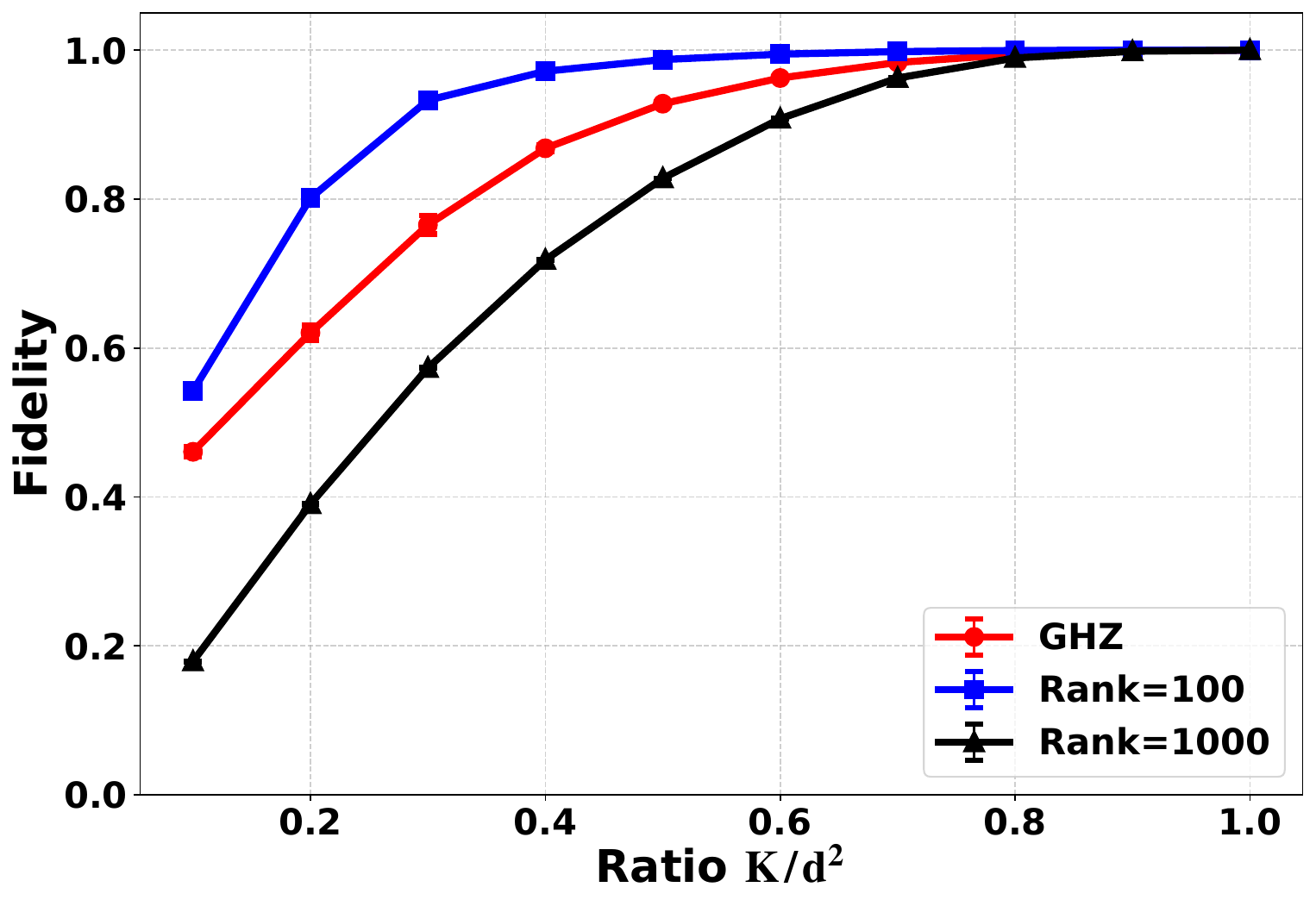}
			\put(-8,65){\textbf{(b)}}
		\end{overpic}
		\begin{overpic}[width=0.70\columnwidth]{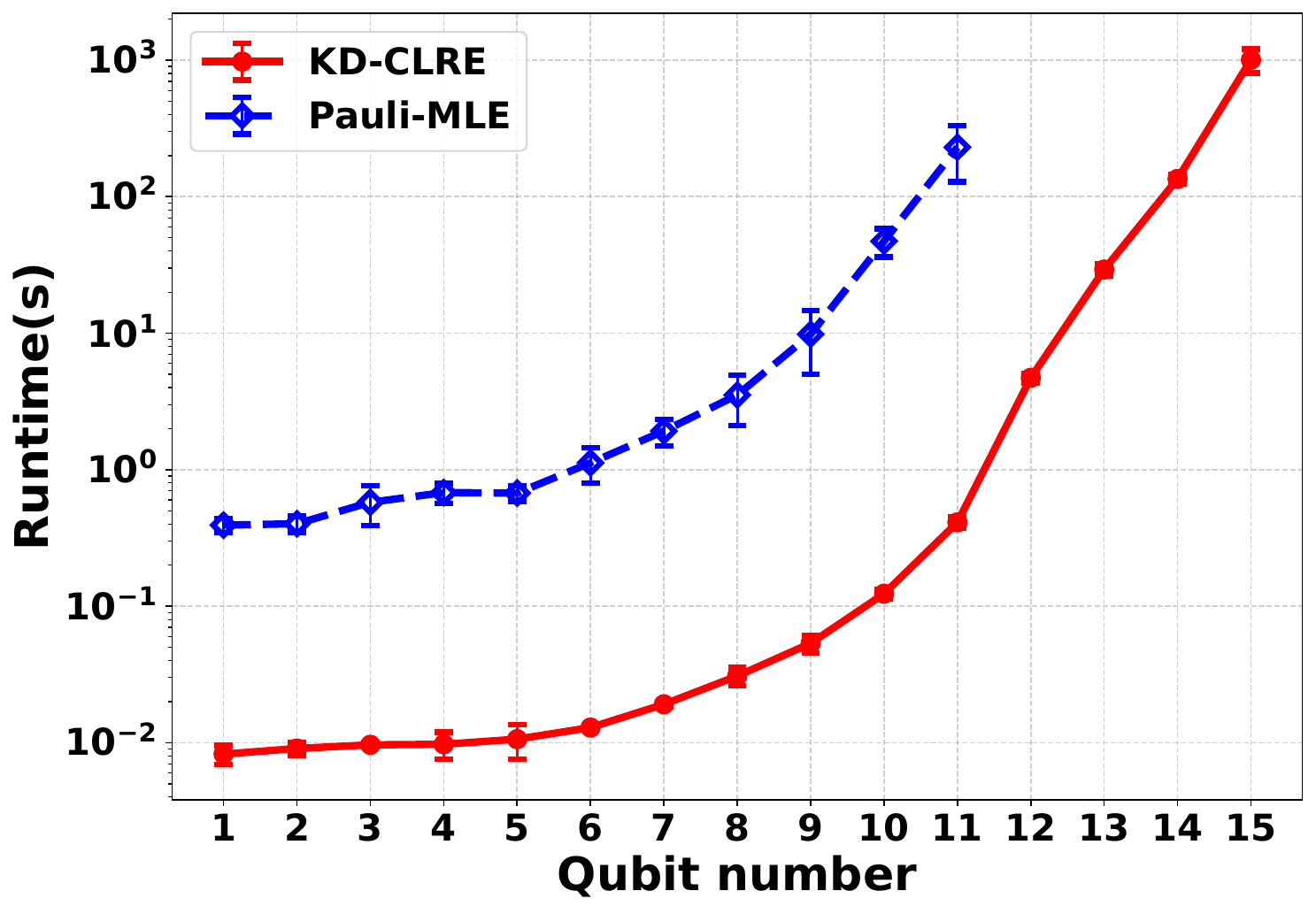}
			\put(-8,65){\textbf{(c)}}
		\end{overpic}
	\end{center}
	\caption{(a) State infidelity versus algorithm iteration for $10$-qubit states.  (b) Reconstruction fidelity of 12-qubit target states versus the used KD-operator ratio $K/d^2$. The curves correspond to a GHZ pure state and random states of rank 100 and 1000.  (c) Runtime for reconstructing states with the KD-based CLRE and with the Pauli-based MLE. Results are obtained by averaging over 20 random states for (a),  10 random datasets for (b) and 10 random states for (c), with error bars denoting standard deviation.}
	\label{fig:convergence runtime}
\end{figure}

 To benchmark our KD-based CLRE, the Pauli-based MLE is used as baseline~\citep{CG-APG}, in which the measurement bases are 
$\{|0\rangle,|1\rangle,|+\rangle,|-\rangle,|+i\rangle,|-i\rangle\}^{\otimes n}$,
with $|{+i}\rangle=1/\sqrt{2}(1,i)^\top$ and $|{-i}\rangle=1/\sqrt{2}(1,-i)^\top$ as eigenstates of $\sigma_y$. The accelerated PGD algorithm with conjugate-gradient initialization  is used to solve the optimization. The step size of Pauli-based MLE is adaptively chosen as in~\citet{Wang2024factored}, whereas the CLRE uses a fixed step size $\alpha=2d^3$. This choice of CLRE is motivated by a simple scaling heuristic: the state update in~(\ref{eqa:MLE iteration}) can be rewritten as  
$\tilde{\rho}_{l+1} = \mathcal{S}\Bigl[\tilde{\rho}_{l}
	- \alpha \|\nabla f_{\rm CLRE}(\tilde{\rho}_{l})\|_F\cdot\frac{\nabla f_{\rm CLRE}(\tilde{\rho}_{l})}{\|\nabla f_{\rm CLRE}(\tilde{\rho}_{l})\|_F}\Bigr]$, 
where $||\cdot||_F$ is the Frobenius norm. The effective step length is $\alpha_0=\alpha \|\nabla f_{\rm CLRE}(\tilde{\rho}_{l})\|_F$, automatically shrinking as the gradient decreases during iteration and thus improving stable convergence. 
The $d^3$ factor compensates for the rapid decay of $\|\nabla f_{\rm CLRE}(\tilde{\rho}_{l})\|_F$, and the prefactor $2$ is chosen empirically as the largest value such that the optimization process remains consistently stable without noticeable oscillations.

Both the MLE and CLRE algorithms are initialized with full-rank $n$-qubit states sampled via~(\ref{mixed state}), and terminated when fidelity change between successive iterations falls below $10^{-5}$. All experiments are performed on a computer equipped with an NVIDIA A100-SXM4 GPU (80G GPU memory and 48G CPU memory).  The code is available on https://github.com/xiang5074/KD-QST.

\subsection{Fast convergence of CLRE with PGD}

We first verify the feasibility of CLRE, by examining the convergence of the PGD algorithm. Particularly, numerical experiments are performed to reconstruct $10$-qubit states from ideal, noise-free data, and the corresponding performance is evaluated by infidelity 
${\rm IF}:= 1-F(\tilde\rho,\rho)$
with the fidelity function $F$ given in~(\ref{fidelity}). 

It is shown in Fig.~\ref{fig:convergence runtime}(a) that our CLRE illustrated as the red solid line achieves infidelity below $10^{-5}$ within 10 iterations, confirming its effectiveness in reconstructing quantum states with high accuracy. By contrast, the convergence of the Pauli-based MLE shown as the blue dashed line requires approximately 1000 iterations to reach an infidelity of $10^{-5}$, much more than that of CLRE. The fast convergence performance of CLRE with PGD arises from numerical stability enabled by the softmax mapping. Since the per-iteration calculation time is comparable for both methods, the fast convergence of CLRE leads to substantial runtime reduction, to be confirmed in the next subsection.

We also investigate its feasibility and performance in the informationally incomplete case, by examining how the reconstruction fidelity changes as the KD-operator ratio $K/d^2$ varies. For each ratio, we randomly use a subset of KD operators to reconstruct 12-qubit states. As shown in Fig.~\ref{fig:convergence runtime}(b), it remains feasible and stable that the reconstruction fidelity increases smoothly with the increasing number of KD operators and exceeds 0.9 at a ratio of 0.7. It is also shown that its performance depends on the target state, such as state rank and entanglement.

\subsection{Runtime speedup for state reconstruction}

\begin{figure}[t]
	\begin{center}
		\begin{overpic}[width=0.70\columnwidth]{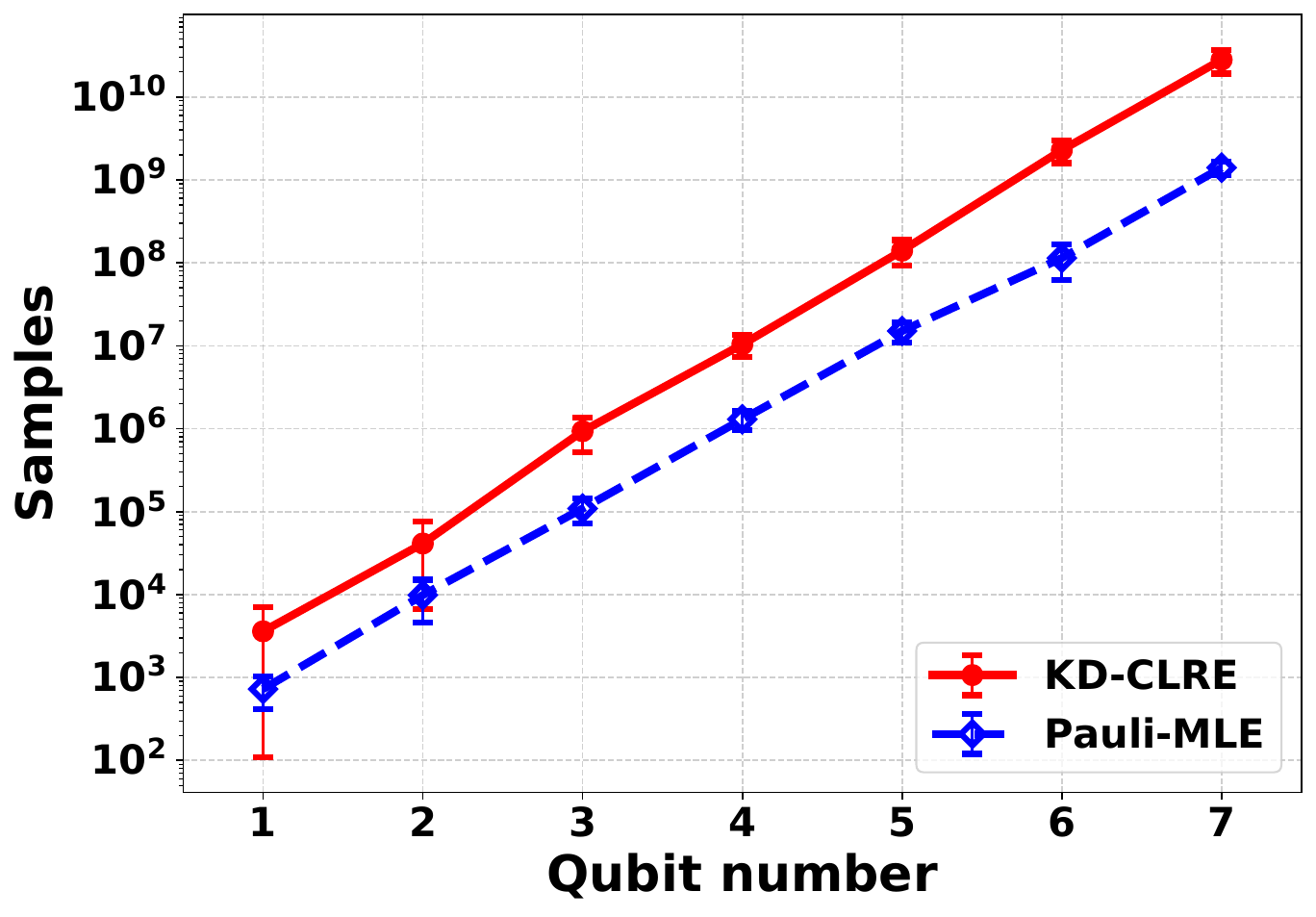}
			\put(-8,65){\textbf{(a)}}
		\end{overpic}
		\begin{overpic}[width=0.70\columnwidth]{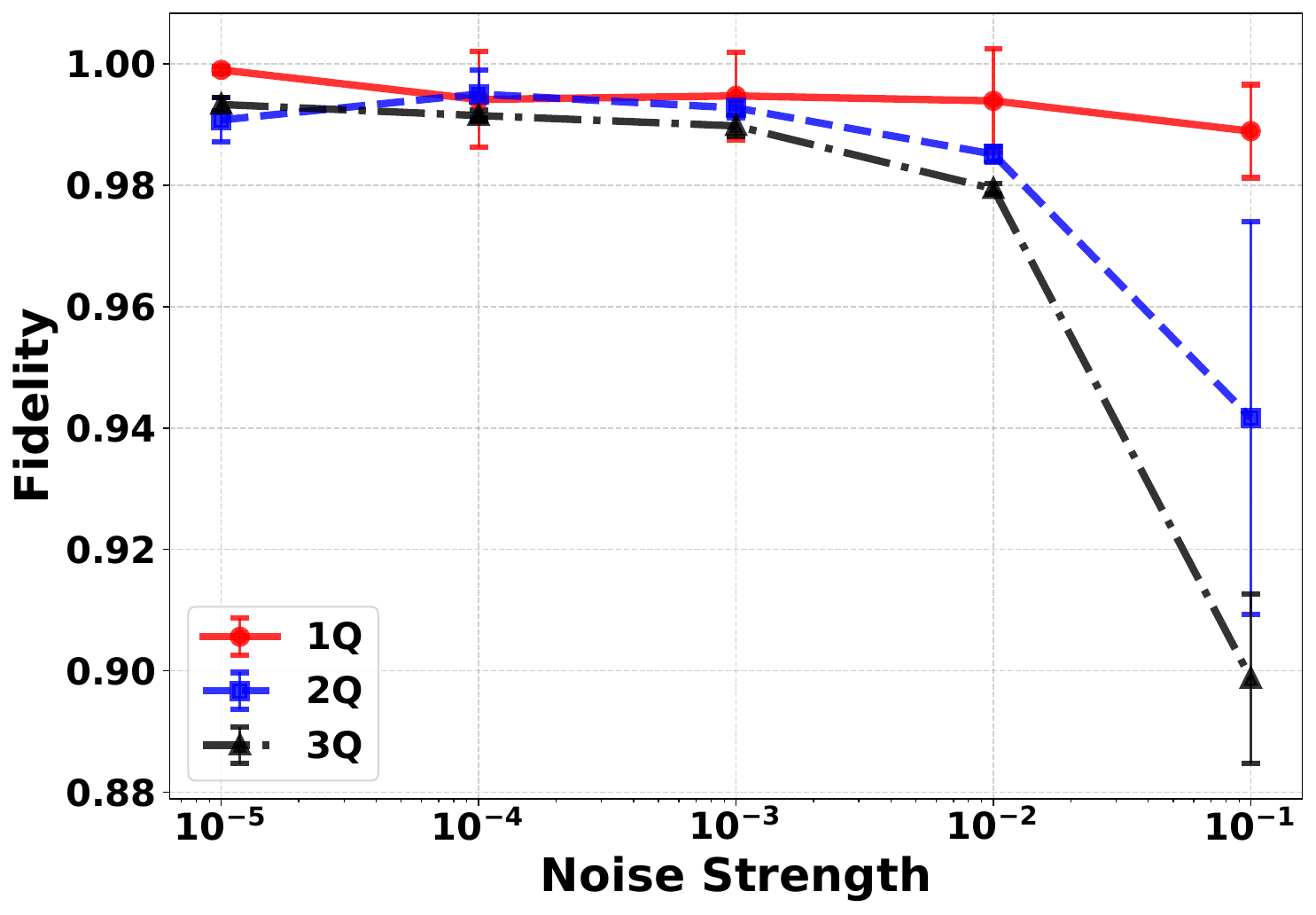}
			\put(-8,65){\textbf{(b)}}
		\end{overpic}
	\end{center}
	\caption{(a) Measurement samples needed to achieve a reconstruction fidelity of 0.99 versus the qubit number. (b) State fidelity of 1-, 2-, and 3-qubit states under depolarizing gate noise of different strengths using the KD-based QST. Each point is averaged over 10 random states, with error bars as standard deviation.}
	\label{fig:noise}
\end{figure}

We then investigate the computational efficiency of CLRE, by examining the  wall-clock  runtime of reconstructing algorithms with varying qubit number. The reported runtime excludes the time cost of state generation and fidelity evaluation. All results are obtained over 10 random states, and error bars indicate the standard deviations. 

As shown in Fig.~\ref{fig:convergence runtime}(c), CLRE achieves an empirical speedup of up to 100-fold compared with the Pauli-based MLE in the tested benchmarks. Under the GPU implementation described in Section~\ref{sec:exp setup}, it accomplishes the reconstruction of 15-qubit randomly generated mixed-state instances within 20 minutes. Moreover, for systems ranging from 1 to 15 qubits, it requires no more than 20 iterations to converge. As each PGD iteration requires one projection, the above convergence bound entails no more than 20 full dense eigen-decomposition calls, which dominate the wall-clock runtime for large systems. We also evaluated a practical stopping criterion, where the algorithm terminates when $\|\rho^{(l)}-\rho^{(l-1)}\|_F<10^{-4}$ for three consecutive iterations. With this criterion, CLRE exhibits a smaller yet still considerable empirical speedup of up to 50-fold compared to Pauli-based MLE.

\subsection{Measurement efficiency}

We next study the measurement overhead of our KD-based QST by examining how many total samples  and observables  are needed to reach a target reconstruction accuracy.  Here, a compact circuit-based protocol is used to generate KD data as proposed by~\citet{Cheng2025} where the target system is coupled to an ancilla of the same dimension and a control qubit, and the joint system then goes through a swap test. The real and imaginary parts of the KD quasiprobabilities can then be obtained by measuring the Pauli observables X and Y on the control qubit, respectively.  

It is shown in Fig.~\ref{fig:noise}(a) that more measurement samples are needed for this two-setting KD protocol to achieve fidelity of 0.99 than the Pauli-based baseline, about one order of magnitude across all system sizes. Despite this overhead, the sampling cost remains comparable and does not negate the advantage of minimal measurement reconfiguration. It is also noted that KD quasiprobabilities could be generated from other experimental routes with different sampling cost, and a comprehensive comparison is left for future work.

\subsection{Noise robustness on the Qiskit platform}

We finally implement experiments on the Qiskit Aer Simulator to evaluate the practical performance of our KD protocol. The depolarizing channel
\begin{equation}
	 \mathcal{D}_{d,\lambda}(\rho):=(1-\lambda)\rho+\frac{\lambda}{d}\mathbb{I}_d
\end{equation}
is used to model real device gate imperfections, where $\lambda$  denotes noise strength. The single-qubit gate noise strength is set as $\lambda=10^{-3}$, while multi-qubit gate noise strength varies from $0$ to $0.1$. Each circuit configuration is executed $10^{n+3}$ times to ensure sufficient statistical accuracy. 
As shown in Fig.~\ref{fig:noise}(b), the reconstruction fidelity remains above $0.98$ for  single qubit states even with a gate-noise strength of $10^{-1}$, and stays around $0.9$ under the same noise level for 3-qubit states. The observed result indicates that the KD-based QST protocol retains good reconstruction performance under noise conditions, making it a potential candidate for implementation on the current quantum devices.

\section{Conclusion}\label{sec:Conclusion}
We have established a novel framework for QST based on KD quasiprobability to reduce measurement burden and computational overhead in previous methods. It is first obtained that two complementary rank-one projective measurements  are sufficient to accomplish full state reconstruction, and then the CLRE estimator, together with the PGD algorithm, is developed to improve numerical stability and speed up convergence. Furthermore, the tensor-product structure of KD operators is utilized to reduce per-iteration computational cost under local product measurements. Finally, extensive numerical experiments are implemented to validate the proposed QST protocol. These findings not only establish the KD-based framework as a promising route for benchmarking large-scale quantum systems, but also open avenues for applying quasiprobability representations to other quantum identification tasks.

There are many interesting questions left for future work. For example, what is the optimal measurement choice for the KD-based QST? In this work, two mutually unbiased bases are tested by experiments without a rigorous proof, as they tend to have better numerical conditioning. It is also interesting to adapt the protocol to specific hardware constraints and to explore its integration with error mitigation techniques.



\bibliographystyle{references}        
\bibliography{refs}           

\end{document}